\documentclass[11pt]{article}
\usepackage[margin=1in]{geometry}
\usepackage{graphicx}
\usepackage{float}
\usepackage{amssymb}
\usepackage{color}
\usepackage{amsthm}
\usepackage{amsmath}
\usepackage{caption}

\usepackage{jasa_harvard}    
\usepackage{JABES_manu}

\renewcommand{\harvardand}

\newtheorem{mydef}{Definition}

\newcommand{\spa}{\mathrm{SP}}
\newcommand{\pr}{\mathrm{Pr}}
\newcommand{\var}{\mathrm{Var}}

\newtheorem{theorem}{Theorem}[section]

\newenvironment{remark}[1][Remark]{\begin{trivlist}
\item[\hskip \labelsep {\bfseries #1}]}{\end{trivlist}}

\begin{document}

\title{Two Sample Order Free Trend Inference with an Application in Plant Physiology}

\author{Yishi Wang \\
Department of Mathematics and Statistics\\ University of North Carolina Wilmington, Wilmington, NC 28409  \\ 
email: \texttt{wangy@uncw.edu} 
\and
Ann E. Stapleton \\
Department of Biology and Marine Biology\\ University of North Carolina Wilmington,Wilmington, NC 28409  \\ 
email: \texttt{stapletona@uncw.edu}
\and 
Cuixian Chen \\
Department of Mathematics and Statistics\\ University of North Carolina Wilmington,Wilmington, NC 28409  \\ 
email: \texttt{chenc@uncw.edu}} 
\maketitle


\newpage
\begin{center}
\textbf{Abstract}
\end{center}
This work is motivated by a biological experiment with a split-plot design, for the purpose of comparison of the changing patterns in seed weight from two treatment groups as subgroups in each of the two groups subject to increasing levels of stress. We formalize the question into a nonparametric two sample comparison problem for changes among the sub samples, which was analyzed using U-statistics. Zero inflated value were also considered in the construction of the U-statistics.  The U-statistics were then used in a Chi-square type test statistics framework for hypothesis testing. Bootstrapped p-values were obtained through simulated samples. It was proven that the distribution of the simulated sample can be independent provided the observed samples have certain summary statistics. Simulation results suggest that the test is consistent.
\vspace*{.3in}

\noindent\textsc{Keywords}: {Sub-sample trend comparison, order free inference,zero inflated value, U-statistics}

\newpage

\section{Introduction}
Biological systems have, as a defining feature, regulation of their processes in ways that generate peaks and dips in measured values over time and space \cite{campbell_integrating_2014}.  Perturbation of biological systems outside the tolerance limits can generate declining growth or activity, and such an out-of-tolerance stress response is also typically non-linear, with thresholds bounding sloped regions.  Common examples of biological stress responses that generate declining measured values include increasing levels of disease agents, toxin exposure \cite{hodgson_textbook_2011} and increasing levels of nutrient deprivation.  In some cases, stress-induced declines can be modified by treatments, such as drugs that counteract toxins or that increase pathogen resistance.  
	Nutrient and crowding stress are especially important in crop growth, with drought being the most economically important limit to crop yield world-wide \cite{lobell_influence_2012}.  Fertilization is frequently used to relieve nitrogen deficiency in grain crops, though nitrogen supplementation comes with both economic cost and side-effects from runoff \cite{hirel_challenge_2007}.  When stresses are applied in combination, as in typical crop growth in fields, the complexity of agronomic and genetic improvement is multiplied.  
	Better understanding of the biological mechanisms that integrate different stress effects during plant growth would allow us to design more focused, less combinatorially complex crop improvement strategies \cite{hammer_models_2006}.  Hormones are common integrators in biological systems,and hormones - by definition - have multiplier effects and control multiple downstream physiological actions \cite{campbell_integrating_2014}.  We designed an experiment to test the effect of plant hormones and hormone-perturbing chemicals on the relative response to environmental stress in maize, a key crop and plant genetic model organism.  In this experiment, as in many biological experiments, the stress effects were applied across the full range from the unstressed normal control to complete lack of growth.  The statistical consequences of this design include a biased zero-inflated data distribution, with most of the zeros in the most severe stress.  In addition, to address our question about possible stress amelioration by chemical treatments, we needed to compare the pattern of growth at each measured point and determine whether growth changed. As this is a real data motivating example, missing values were accommodated in our new analysis.  

In the experiment, multiple maize inbreds were exposed to all combinations of the following stressors: drought, nitrogen, and density stress. Plants were grown in an experimental plot divided into eight sections, and each of the sections received a combination of between zero and three of the stresses previously mentioned, so that all possible stress combinations were included.  More details about the experiment can be found in \cite{stutts_hormones_2014}.

The following boxplot shows the seed weight response changes as the environment(ENV) levels vary, for chemical treatments " PAC" and "PACGA" in genotype Mo298. The EVN levels can be ordered according to total stress levels across the individual and combination stresses.  One question to ask is whether these two boxplots exhibit the same pattern in the change of seed weight as the ENV level changes.

\begin{figure}[h]
\center
\includegraphics[width=12cm]{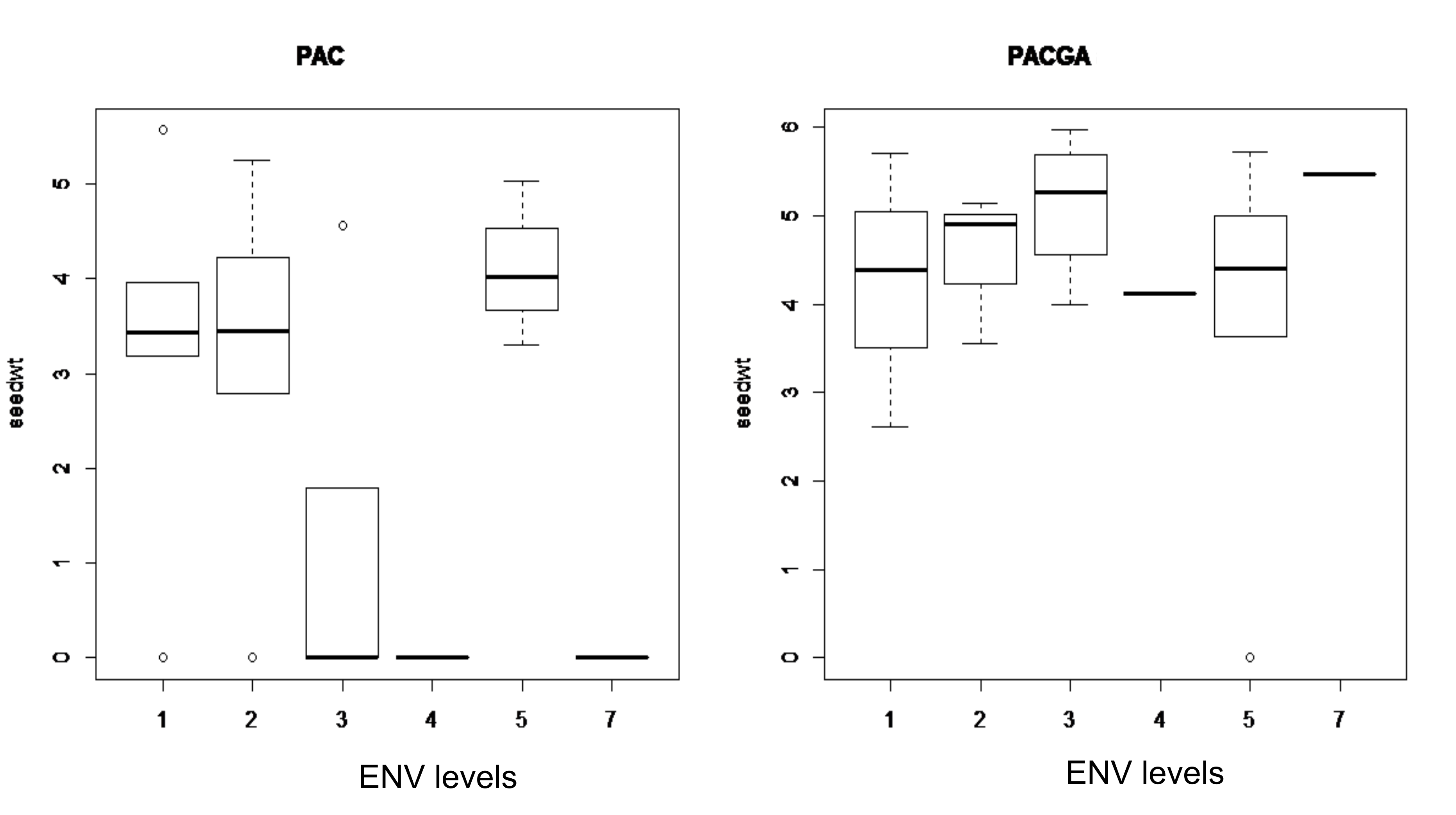}
\caption{Boxplot of sub samples for PAC and PACGA from Line Mo298}\label{exp1.fig}
\end{figure}

Motivated by the above example, our focus in this work is to develop a methodology for trend comparisons among sub-samples for the two sample situation. To be specific within this context, we introduce the following notation.

First, we assume that there are two populations that are subject to two different treatments, denoted by $\mathbb{A}$ and $\mathbb{B}$. For population $\mathbb{A}$, there are $K+1$ sub populations with inherent orders. Let $X_{ij}$ be the $j$-th random replicates from the $i$-th sub population, for $i=1,\cdots,K+1$ and $j=1,\cdots,n_{i}$, where $n_{i}$'s are the number of replications in each subgroup. Let $F_i(\theta_{x,i})$ be the cumulative distribution function (cdf) of $X_{ij}$, where  $\theta_{x,i}$'s are the parameters.  For population $\mathbb{B}$, there are also $K+1$ sub populations with inherent orders. Let $Y_{ij}$ be the $j$-th random replicates from the $i$-th sub population, for $i=1,\cdots,K+1$ and $j=1,\cdots,m_{i}$, where $m_{i}$'s are the number of replications in each subgroup. Let $G_i(\theta_{y,i})$ be the cumulative distribution function (cdf) of $Y_{ij}$, where  $\theta_{y,i}$'s are the parameters. The trend comparison problem among two samples could possibly be modeled to test whether the pattern among the $\theta_{x,i}$'s is the same with the one among the $\theta_{y,i}$'s.

Based on the above notation, many statistical models can be viewed as a one sample statistics problem. For monotone trend detection, the Mann-Kendall test \cite{mann1945nonparametric} \cite{kendall1948rank} \cite{richard1987statistical} was proposed to examine whether such trend existed for one variable of interest, with $n_i=1$ (without sub samples structure). Another example of non-parametric trend comparison \cite{hirsch1984nonparametric}  used rank transformation, which is computationally tractable but does not relax the assumption of equal variances. Trend comparison statistical approaches often rely primarily on distributional assumptions.  Linear fits are often used in small time series settings, with longer time series data allowing applications of additional methods such as shifted correlations and cycles \cite{brockwell2006introduction}. 

Analysis of Variance is an alternative and widely applied approach; under the current notation its hypothesis becomes 
\begin{align}
H_0: \theta_{x,1}=\cdots=\theta_{x,K+1}\text{ v.s. }H_1: \theta_{x,i}\neq\theta_{x,j} \text{ for } i\neq j\in \{1,2,\cdots,K+1\},\nonumber
\end{align}
where $\theta_{x,i}$ is the $i$th sub population average from population $\mathbb{A}$. Along these lines, order-restricted inference, which appeared first in \cite{bartholomew1959test}, \cite{bartholomew1959test2}, was among the earliest works with an alternative hypothesis on the monotonic pattern of the subgroup parameters:
\begin{align}
H_1: \theta_{x,1}\leq\cdots\leq\theta_{x,K+1}\text{ or } \theta_{x,1}\geq\cdots\geq\theta_{x,K+1},\nonumber
\end{align}
with at least one strict inequality.

Among recent work in order restricted inference, a broader class of order restrictions were investigated \cite{robertson1988order}. For binary type response variables, a test statistic was proposed for testing equality of multiple independent variables against three types of ordered restrictions\cite{peddada2001tests}. Methodologies were also developed for time-course or dose-response profiles with gene expression data \cite{peddada2003gene}\cite{peddada2010analysis}\cite{simmons2007order}. Another related work is \cite{deng2000detecting}, in which a test based on Mann-Whitney statistics was proposed to test hormesis of dose response relationship in one-way layout design. Multiple independent random sub samples were involved for testing against monotone or hormesis orders among the sub population averages.

All the literature that we have found has a focus on patterns of parameters from sub populations of one population. It is therefore appealing to investigate the comparison of parameters' patterns (from sub populations) of two populations. Other considerations include relatively small sub-sample sizes (~5) and zero inflated values. Small sample size makes any distributional assumption hard to verify. Non-parametric analysis allows us to avoid the assumption that the experimental samples were chosen from a specified distribution; robust and resampling-based non-parametric tests also free us from the assumption of equal variances \cite{quinn2002experimental}. Non-parametric tests would be especially useful for the typical small-n short series commonly utilized in biological experimental designs - if such tests were available for trends.Therefore we will consider nonparametric approach in this work. 

\section{Methodology}

To investigate the trend among sub-samples, we need to first define what we mean by trend or sequential pattern. 

\begin{mydef}\label{def1}
The sequential pattern of a sequence of continuous numbers, $\alpha_1,\cdots,\alpha_k$ is defined as $$\spa(\alpha_1,\cdots,\alpha_k)\triangleq \{1(\alpha_1<\alpha_2),1(\alpha_2<\alpha_3),\dots,1(\alpha_{k-1}<\alpha_k)\},$$ 
where $1(\cdot)$ is the indicator function. 
\end{mydef}

\begin{mydef}\label{def2}
The sequential patterns of two sequence of continuous numbers, $\alpha_1,\cdots,\alpha_k$ and $\beta_1,\cdots,\beta_k$ are equivalent if and only if $\spa(\alpha_1,\cdots,\alpha_k)=\spa(\beta_1,\cdots,\beta_k)$.
\end{mydef}

We defined $p_{x,i}\triangleq \pr(X_{i1}<X_{(i+1)1})$ and $p_{y,i}\triangleq \pr(Y_{i1}<Y_{(i+1)1})$, then the hypothesis testing question is 
\begin{align}\label{hyp}
H_0&:(p_{x,1},p_{x,2},\cdots,p_{x,K})=(p_{y,1},p_{y,2},\cdots,p_{y,K}),\nonumber \\ 
\text{ v.s. } H_1&:p_{x,i}\neq p_{y,i} \text{ for at least one }i\in \{1,\cdots, K\}.
\end{align} 

When $K=1$, the above hypotheses appears to be the same as the usual two-sample proportion test. However in the context of this paper, the proportion $p_{x,1}$ and $p_{y,1}$ are each estimated from two independent samples, and thus we can not use the regular two sample proportion test to evaluate the hypotheses. 
Assume that we have sample observations $x_{ij}$ for $i=1,\cdots,K+1$ and $j=1,\cdots,n_{i}$, and $y_{ij}$ for $i=1,\cdots,K+1$ and $j=1,\cdots,m_{i}$. 
In order to estimate $p_{x,1}$, a convenient estimator would be  
\begin{align}\label{prop.est0}
\hat{p}_{0,x,1}\triangleq\frac{1}{n_0}\sum_{j=1}^{n_0}1(x_{1j}<x_{2j}),
\end{align} 
if $n_0=n_1=n_2$.

However, the sub sample sizes from two different sub populations may be different, we therefore propose the following estimator:
\begin{align}\label{prop.est}
\hat{p}_{x,l}\triangleq\frac{1}{n_ln_{l+1}}\sum_{i=1}^{n_l}\sum_{j=1}^{n_{l+1}}1(x_{li}<x_{(l+1)j}),
\end{align} 
for $l=1,\cdots,K$. This estimator is based on Mann-Whitney $U-$ statistics \cite{mann1947test}. The strength of this later estimator is not just limited to its flexibility  to sub-sample size issue. It can be verified that E$(\hat{p}_{x,1})=p_{x,1}$ and  $\var(\hat{p}_{x,1})=\frac{p_{x,1}(1-p_{x,1})}{n_1n_2}$, which is  more efficient than the estimator in (\ref{prop.est0}).

Similarly, the estimator for $p_{y,l}$ is:
\begin{align}\label{prop.est.y}
\hat{p}_{y,l}\triangleq\frac{1}{m_lm_{l+1}}\sum_{i=1}^{m_l}\sum_{j=1}^{m_{l+1}}1(y_{li}<y_{{(l+1)}j}).
\end{align}

Following the idea of the estimator in (\ref{prop.est}), we propose to use the $2\times 2K$ table below to summarize information for sub-sample trend comparison:

\begin{table}[h]
\centering
\begin{tabular}{ |c|c|c|c|c| } 
 \hline
  $O_{x,1}$ &  $n_1n_2-O_{x,1}$ & $\cdots$  &$O_{x,K}$ &$n_{K}n_{K+1}-O_{x,K}$  \\ 
 \hline
 $O_{y,1}$ &  $m_1m_2-O_{y,1}$ & $\cdots$  &$O_{y,K}$ &$m_{K}m_{K+1}-O_{x,K}$  \\
 \hline
\end{tabular}
\caption{Frequency Distribution Table}\label{tab:fdt}
\end{table}
\noindent where $O_{x,l}\triangleq\displaystyle\sum_{i=1}^{n_l}\sum_{j=1}^{n_{l+1}}1(x_{li}<x_{(l+1)j})$  and $O_{y,l}\triangleq\displaystyle\sum_{i=1}^{m_l}\sum_{j=1}^{m_{l+1}}1(y_{li}<y_{(l+1)j})$, for $l\in{1,\cdots,K}$.

The first two columns are for the purpose of comparing $p_{x,1}$ and $p_{y,1}$
 with row sums equal to $n_1n_2$ for the first row and $m_1m_2$ for the second row. Therefore it follows that the sum of the entire first row is $n_1n_2+n_2n_3+\cdots+n_{K}n_{K+1}$  and the sum of the entire second row is $m_1m_2+m_2m_3+\cdots+m_{K}m_{K+1}$. Let $N_{tot}\triangleq\sum_{l=1}^K(n_ln_{l+1}+m_lm_{l+1})$ and $R_x\triangleq \frac{\sum_{l=1}^K(n_ln_{l+1})}{N_{tot}}$, as these quantities will be helpful to develop the expected frequency table below.
 
\begin{theorem}\label{exp.table}
With Table\ref{tab:fdt} and under the null hypothesis in (\ref{hyp}), the expected frequency table follows:

\begin{table}[h]
\centering
\begin{tabular}{ |c|c|c|c|c| } 
 \hline
  $E_{x,1}$ &  $R_x(n_1n_{2}+m_1m_{2})-E_{x,1}$ & $\cdots$  &$E_{x,K}$ &$R_x(n_Kn_{K+1}+m_Km_{K+1})-E_{x,K}$  \\ 
 \hline
 $E_{y,1}$ &  $(1-R_x)(n_1n_{2}+m_1m_{2})-E_{y,1}$ & $\cdots$  &$E_{y,K}$  &$(1-R_x)(n_Kn_{K+1}+m_Km_{K+1})-E_{y,K}$  \\
 \hline
\end{tabular}
\caption{Expected Frequency Table under $H_0$}\label{tab:exp}
\end{table}
\noindent where $E_{x,l}\triangleq(O_{x,1}+O_{y,1})R_x$ , and $E_{y,l}\triangleq(O_{x,1}+O_{y,1})(1-R_x)$.
\end{theorem}
\begin{proof}
The row sum of the first row of Table\ref{tab:fdt} is $\sum_{l=1}^K(n_ln_{l+1})$, and the row sum of the second row is $\sum_{l=1}^K(m_lm_{l+1})$. The sum of the first column is $O_{x,1}+O_{y,1}$. By the definition of independence, we have expected count for the first cell of the first row as

\begin{align}\label{thm1.1}
E_{x,1}=R_x\times\frac{(O_{x,1}+O_{y,1})}{N_{tot}}\times N_{tot}=(O_{x,1}+O_{y,1})R_x.
\end{align}

Similarly, we can verify that  $E_{y,1}=(O_{x,1}+O_{y,1})(1-R_x)$.

As to the second cell of the first row, noting that the sum of the second column is $(n_1n_{2}+m_1m_{2}-O_{x,1}-O_{y,1})$, we have the expected value of the that cell to be
\begin{align}\label{thm1.2}
R_x\times\frac{(n_1n_{2}+m_1m_{2}-O_{x,1}-O_{y,1})}{N_{tot}}\times N_{tot}=R_x\times(n_1n_{2}+m_1m_{2})-E_{x,1}.
\end{align}
Similarly, we can verify that  the expected value of the second cell of second row is $(1-R_x)(n_1n_{2}+m_1m_{2})-E_{y,1}$. The proof for the rest of the cells can be verified similarly.

\end{proof}

Following the idea of a chi-square test, by comparing Table\ref{tab:fdt} with Table\ref{tab:exp}, the following statistic is constructed to measure the discrepancy between $\hat{p}_{x,l}$ and $\hat{p}_{y,l}$:

\begin{align}\label{test}
M\triangleq\sum_{l=1}^{K}&\left(\frac{(O_{x,l}-E_{x,l})^2}{E_{x,l}}+\frac{(n_ln_{l+1}-O_{x,l}-R_x(n_ln_{l+1}+m_lm_{l+1})+E_{x,l})^2}{(R_x(n_ln_{l+1}+m_lm_{l+1})-E_{x,l})}\right)\nonumber\\
+\sum_{l=1}^{K}&\left(\frac{(O_{y,l}-E_{y,l})^2}{E_{y,l}}+\frac{(m_lm_{l+1}-O_{y,l}-(1-R_x)(n_ln_{l+1}+m_lm_{l+1})+E_{y,l})^2}{\left((1-R_x)(n_ln_{l+1}+m_lm_{l+1})-E_{y,l}\right)}\right).
\end{align} 
\begin{remark}[Remark 1:] The denominators $E_{x,l}$ in (\ref{test}) can be zero if and only if the corresponding observed counts are zero, for both $x$ and $y$. $E_{x,l}$ as an example, from (\ref{thm1.1}), is zero if and only if both $O_{x,l}$ and $O_{y,l}$ are zeroes. When $R_x(n_ln_{l+1}+m_lm_{l+1})-E_{x,l}=0$, by (\ref{thm1.1}), we have $(n_ln_{l+1}+m_lm_{l+1})=(O_{x,l}+O_{y,l})$, which is equivalent with $n_ln_{l+1}=O_{x,l}$ and $m_lm_{l+1}=O_{y,l}$, since $n_ln_{l+1}\geq O_{x,l}$ and $m_lm_{l+1}\geq O_{y,l}$. Thus its corresponding observing cell $n_ln_{l+1}-O_{x,l}$ is zero. In either of the two cases (when $O_{x,l}= O_{y,l}=0$; and when $n_ln_{l+1}=O_{x,l}$ and $m_lm_{l+1}=O_{y,l}$), we have $\frac{0}{0}$. Both cases also indicate that the cell counts of one column are zero, which means that for both samples, each observation in one sub-sample is always less than each observation in the other sub-sample. Since the two sample patterns agree, it is natural to define  $\frac{0}{0}\triangleq 0$ when such fractions exist for the test statistic $M$.
\end{remark}

\begin{remark}[Remark 2:] In the case of comparing zero inflated values, say $u$ from one sub-sample and $v$ from another sub-sample, we may replace the indicator function $1(\cdot)$ in $M$, by using the following function:
\begin{align}\label{censor}
f(u,v)=X_B,
\end{align}
\noindent where $X_B$ follows Bernoulli distribution with $p=\frac{1}{2}$.
\end{remark}


\section{Resampling Scheme for Hypothesis Testing and Power Analysis}
One way of finding critical values for test statistics $M$ in (\ref{test}) is through developing the asymptotic distribution. This approach however may require relatively large sub-sample sizes $n_l$ and $m_l$. Hence we propose a pseudo approach using resampling.

Resampling methods have been considered as an effective approach to simulate critical values as well as powers of tests. A comprehensive review on the subject can be found in \cite{efron1994introduction} and \cite{davison1997bootstrap}. Traditional resampling methods resample or relabel observations from the obtained samples, with or without replacement. Bootstrap methods were adopted in order restricted inference  \cite{peddada2001tests}. However, such an approach is not feasible in this context, because the parameters in the null and alternatives hypotheses could not be conveniently related to the sub sample observations. 

Specifically, the parameters we have are: $p_{x,1},p_{x,2},\cdots,p_{x,K}$ and $p_{y,1},p_{y,2},\cdots,p_{y,K}$. Under the null hypothesis in (\ref{hyp}), we have $(p_{x,1},p_{x,2},\cdots,p_{x,K})=(p_{y,1},p_{y,2},\cdots,p_{y,K})$. Define
\begin{align}\label{pl}
p_l\triangleq p_{x,l}=p_{y,l}.
\end{align}
Take $p_{x,1}$ and $p_{y,1}$ as an example, after obtaining sample estimators $\hat{p}_{x,1}$ and $\hat{p}_{y,1}$ following (\ref{prop.est}) and (\ref{prop.est.y}), an unbiased estimator for $p_{x,1}=p_{y,1}$ under the null hypothesis in (\ref{hyp}) is:
\begin{align}\label{cbnest.1}
\hat{p}_1\triangleq \frac{n_1n_2\hat{p}_{x,1}+m_1m_2\hat{p}_{y,1}}{n_1n_2+m_1m_2}.
\end{align}
It is not an easy task to resample or relabel observations from $\{x_{1i},x_{2j},y_{1k},y_{2l} \}_{i,j,k,l=1}^{m_1,m_2,n_1,n_2}$, such that with $\hat{p}_1$ as defined in (\ref{cbnest.1})
\begin{align}\label{bootpr.1}
\pr(X_{1,i}^*<X_{2,j}^*)=\pr(Y_{1,k}^*<Y_{2,l}^*)=\hat{p}_1,    
\end{align}
in which the sign $^*$ indicates that the random variable is bootstrapped from the sample observations.

Instead of resampling or relabeling observations from the existing observations, we propose to randomly generate $x_{li}^*$'s and $y_{lj}^*$'s from underline distributions $F^*_l(\eta_l)$, where $\eta_l$ is the parameter, such that for i.i.d random variables $X_{l}^*\sim F^*_l(\eta_l)$ and $Y_{l}^*\sim F^*_l(\eta_l)$,
\begin{align}\label{bootpr}
\pr(X_{l}^*<X_{l+1}^*)=\pr(Y_{l}^*<Y_{l+1}^*)=\hat{p}_l,    
\end{align}
for $l=1,\cdots,K$, where
\begin{align}\label{cbnest}
\hat{p}_l\triangleq \frac{n_ln_{l+1}\hat{p}_{x,l}+m_lm_{l+1}\hat{p}_{y,l}}{n_ln_{l+1}+m_lm_{l+1}}.
\end{align}

\begin{theorem}\label{thm3.1}
Under the null hypothesis in (\ref{hyp}), given $p_l$ defined in (\ref{pl}),  the distribution of the test statistics $M$ in (\ref{test}) is independent of  $F_l(\theta_{x,l} )$ and $G_l(\theta_{y,l} )$.
\end{theorem}
\begin{proof}
Due to the Chi-square structure of the test statistics $M$, it is sufficient to work with  the special case of $K=1$, which make the distribution of $M$ depending on $O_{x,1}$ and $O_{y,1}$. If it can be proven that given $p_1$, $O_{x,1}$ is independent of $F_1(\theta_{x,1} )$ and $F_2(\theta_{y,1} )$, the theorem follows.

If the two data sequence $\{x_{1i}\}_{i=1}^{n_1}$ and $\{x_{2j}\}_{j=1}^{n_2}$ are combined and ordered from the smallest to the largest, $O_{x,1}$ is essentially counting how many times $x_{1,i}$'s appear before $x_{2,j}$'s.

With $n_1$ observations in the first sub-sample and $n_2$ observations in the second sub-sample, $O_{x,1}$ can be calculated from a sample space that is made of $\frac{(n_1+n_2)!}{n_1!n_2!}$ ways of ordering the combined sample. The chance that the largest observation coming from the first sample is $\frac{n_1}{n_1+n_2}$ -- in which case, if we delete that largest observation, $O_{x,1}$ is intact; on the other hand, the chance that the largest observation coming from the second sample is $\frac{n_2}{n_1+n_2}$ -- in which case, if we delete that largest observation, $O_{x,1}$ is reduced to $O_{x,1}-n_1$, since we lost all the counts of the $n_1$ observations less than the deleted one. Let $\pr_{n_1,n_2}(O_{x,1}=k)$ represent the probability of $O_{x,1}=k$, which is non negative when $k$  is a non negative integer and zero elsewhere. Based on the discussions so far, given that the size of the two sub-samples are $n_1$ and $n_2$, we have
\begin{align}\label{recur}
\pr_{n_1,n_2}(O_{x,1}=k)=\frac{n_1}{n_1+n_2}\pr_{n_1-1,n_2}(O_{x,1}=k)+\frac{n_2}{n_1+n_2}\pr_{n_1,n_2-1}(O_{x,1}=k-n_1).
\end{align}
The significance of equation (\ref{recur}) is that $\pr_{n_1,n_2}(O_{x,1}=k)$ is equivalent with the addition of two probabilities, whose "effective" sub-sample sizes are reduced. This recurrence relationship in (\ref{recur}) depends on eventually on positive probabilities when one of the "effective" sub-sample sizes is one:
\begin{align}\label{recur.basic}
\pr_{1,t}(O_{x,1}=s)&=\pr_{t,1}(O_{x,1}=s)=\begin{pmatrix}t \\ s \end{pmatrix}{p_1}^s(1-p_1)^{t-s},
\end{align}
for $s \in \{0,1,\cdots,t\}$, where $t$ is a positive integer less than or equal to $n_1$ or $n_2$. Thus, it is proven that, given $p_1$, $\pr_{n_1,n_2}(O_{x,1}=k)$ in (\ref{recur}) is independent of $F_1(\theta_{x,1} )$ and $F_2(\theta_{x,1} )$. Similar discussions can be shown that, given $p_1$, $\pr_{m_1,m_2}(O_{y,1}=k)$ is independent of $G_1(\theta_{y,1} )$ and $G_2(\theta_{y,1} )$ The conclusion of the theorem follows.


\end{proof}

Theorem\ref{thm3.1} suggests that the distribution of $M$ only depends on $p_l$, which allows free choice of underlying distributions for $F$ and $G$, as long as  (\ref{bootpr}) is satisfied. The normal distribution family is considered in this work, because according to L\'evy characterization (see e.g. Theorem 20.2.A in \cite{loeve1977probability}), linear combinations of independent normal random variables follow a normal distribution.

Let $X_1\sim N(a,1)$ and $X_2\sim N(a+h,1)$ be independent random variables. Then $\pr(X_1<X_2)=\pr(Z<\frac{h}{\sqrt{2}})$, where $Z$ is a standard normal r.v. Therefore, $h=\sqrt{2}\times \Phi^{-1}(\pr(X_1<X_2))$, where $\Phi$ is the cdf of $Z$. Following this idea, in order to generate random samples $x_{li}^*$'s  following the relationship in (\ref{bootpr}), we may set 
\begin{align}\label{hl}
 h_1\triangleq 0, \text{ and } h_i\triangleq\sqrt{2} \Phi^{-1}(\hat{p}_{i-1}),\text{ for } i=2,\cdots, K+1,
\end{align}
and generate random samples $x_{li}^*$'s from $X_i^*\sim N(h_{i},1)$ for $i=1,\cdots,K+1$.
Under the null hypothesis, random samples $y_{lj}^*$'s should be generated in the same way as $x_{li}^*$'s.

The following table summarize all the steps it takes to conduct the hypothesis testing as in (\ref{hyp})

\begin{table}[h]
\centering
\begin{tabular}{ |p{1cm}|p{14.2cm}|} 
 \hline
   Step1 & With the current sample observations $x_{ij}$  for $i=1,\cdots, K+1$ and $j=1,\cdots, n_i$; and $y_{ij}$ for $i=1,\cdots, K+1$ and $j=1,\cdots, m_i$, find $\hat{p}_{x,l}$ as defined in (\ref{prop.est}) and $\hat{p}_{y,l}$ as defined in (\ref{prop.est.y}), for $l=1,\cdots, K$.\\
   Step2 & Calculate $\hat{p}_l$ as defined in ($\ref{cbnest}$) for $l=1,\cdots,K$.\\
   Step3 & Calculate $h_i$ as defined in ($\ref{hl}$) for $i=1,\cdots,K+1$.\\
  Step4 & Generate random samples $x_{ij}^*$'s from $N(h_{i},1)$ for $i=1,\cdots,K+1$ and $j=1,\cdots, n_i$ .\\
   Step5 & Generate random samples $y_{ij}^*$'s from $N(h_{i},1)$ for $i=1,\cdots,K+1$ and $j=1,\cdots, m_i$.\\
   Step6 & Compute $M^*$ based on $x_{ij}^*$'s and $y_{ij}^*$'s with the same formula as defined in (\ref{test}).\\
  Step7 & Repeat step 4 through step 6 $N_b$ many times and obtain a sequence $\{M^*_t\}_{t=1}^{N_b}$.\\
   Step8 & Order the sequence $\{M^*_t\}_{t=1}^{N_b}$ and find its (1-$\alpha$)th percentile, which is the critical value, denoted by $C_{\alpha}$. \\
   Step9 & Compare the test statistics $M$ from (\ref{test}) with $C_{\alpha}$. If $M\leq C_{\alpha}$, the null hypothesis will not be rejected. A bootstrapped p-value may also be obtained through determining the percentage of values in $\{M^*_t\}_{t=1}^{N_b}$ more than $M$. \\
\hline
\end{tabular}
\caption{A summary of the steps for bootstrapping the critical values with type I error $\alpha$.}\label{tab:step}
\end{table}


\section{Simulation Study}
In the first simulation study, the power of the test was studied with varying designs of frequency distribution tables and sample sizes. We assume that there are four sub-samples from each of the population $\mathbb{A}$ and of population $\mathbb{B}$. As shown in Table\ref{tab:pow}, the sub-tables in the third column are realizations of Table\ref{tab:fdt}:

\begin{table}[h]
\centering
\begin{tabular}{ |c|c|l|c|c| } 
 \hline
 ID & Sub-sample Sizes & Freq Distribution Table & p-value & $N_b$\\
 \hline
 1 & \begin{tabular}{@{}l@{}@{}}  $K=3$\\$n_1=n_2=n_3=n_4=5$\\  $m_1=m_2=m_3=m_4=5$  \end{tabular} & \begin{tabular}{ |c|c|c|c|c|c| } 
 \hline
  20 &  5 & 10  & 15 &20 &5  \\ 
 \hline
  15 &  10 & 15  & 10 &20 &5  \\ 
 \hline
\end{tabular} & 0.764 & $10^3$   \\ 
 \hline
    2 & \begin{tabular}{@{}l@{}@{}}  $K=3$\\$n_1=n_2=n_3=n_4=10$\\  $m_1=m_2=m_3=m_4=10$ \end{tabular} & \begin{tabular}{ |c|c|c|c|c|c| } 
 \hline
  80 &  20 & 40  & 60 &80 &20  \\ 
 \hline
  60 &  40 & 60  & 40 &80 &20  \\ 
 \hline
\end{tabular} & 0.449 & $10^3$   \\ 
\hline
     3 & \begin{tabular}{@{}l@{}@{}}  $K=3$\\$n_1=n_2=n_3=n_4=5$\\  $m_1=m_2=m_3=m_4=5$ \end{tabular} & \begin{tabular}{ |c|c|c|c|c|c| } 
 \hline
  18 &  7 & 12  & 13 &22 &3  \\ 
 \hline
  15 &  10 & 15  & 10 &20 &5  \\ 
 \hline
\end{tabular} & 0.921 & $10^3$   \\ 
\hline
     4 & \begin{tabular}{@{}l@{}@{}}  $K=3$\\$n_1=n_2=n_3=n_4=10$\\  $m_1=m_2=m_3=m_4=10$ \end{tabular} & \begin{tabular}{ |c|c|c|c|c|c| } 
 \hline
  72 &  28 & 48  & 52 &88 &12  \\ 
 \hline
  60 &  40 & 60  & 40 &80 &20  \\ 
 \hline
\end{tabular} & 0.747 & $10^3$   \\ 
 
 \hline
\end{tabular}
\caption{Examples of Power Simulation}\label{tab:pow}
\end{table}
For the first experiment in Table\ref{tab:pow}, the sub-sample sizes are uniformly equal to 5. An example Freq Distribution Table is then provided  following the structure of  Table\ref{tab:fdt}. It also indicates that 
\begin{align}\label{ps:exp1}
&\hat{p}_{x,1}=\frac{20}{25}=0.8\text{, }\hat{p}_{x,2}=\frac{10}{25}=0.4 \text{ and } \hat{p}_{x,3}=0.8; \nonumber\\
&\hat{p}_{y,1}=0.6\text{, }\hat{p}_{y,2}=0.6 \text{ and } \hat{p}_{y,3}=0.8, \nonumber\\
&\text{which implies that: } \hat{p}_{1}=0.7\text{, } \hat{p}_{2}=0.5 \text{ and } \hat{p}_{3}=0.8. 
\end{align}
following the steps in Table\ref{tab:step}, a simulation p-value of 0.764 was obtained with  $N_b=10^3$. Therefore with $\alpha=0.05$, it is unlikely to reject the null hypothesis as in (\ref{hyp}).

In the second experiment, we assume that there are ten observations in each of the four sub-samples from population $\mathbb{A}$, as well as from population $\mathbb{B}$. Compared with the first experiment, the sub-sample size in the second experiment is twice as big. The values in the corresponding Freq Distribution Table are exactly four times the values in the Freq Distribution Table values of the first experiment. Hence, $\hat{p}_{x,l}$, $\hat{p}_{x,l}$ and $\hat{p}_{l}$ are exactly the same as in (\ref{ps:exp1}) from the first experiments. With $N_b=10^3$, a p-value of 0.449 was obtained.

By comparing the p-values from the first two experiments in Table\ref{tab:pow}, it suggests that the rejection power of the test is getting bigger as the sub-sample sizes increase.

The third and the fourth experiments were designed using the same structure as the first two experiments in terms of sub-sample sizes and counts in the frequency distribution tables, and therefore they share the same $\hat{p}_l$'s.The same decreasing order of the p-values are observed, from 0.921 to 0.747 as the sub-sample sizes increase.

Another comparison is between the first and third experiments. Among the two experiments, $O_{y,i}$'s are the same, while $O_{x,i}$'s from the third experiment are more similar with $O_{y,i}$'s than the  $O_{x,i}$'s from the first experiment. Therefore it makes sense that the p-value obtained from the third experiment (0.921) is greater than the one (0.764) from the first experiment. Similar patterns among p-values are also observed between the fourth (0.747) and the second (0.449) experiments.

In the second simulation study, the type I error of the test was investigated with $\alpha=0.05$. 
\begin{table}[h]
\centering
\begin{tabular}{ |c|c|c|c|c| c|} 
 \hline
 ID & Sub-sample Sizes & $p_l$ & Err &$N_b$ &$ N_{rep}$\\
 \hline
 1 & \begin{tabular}{@{}l@{}@{}}  $K=3$\\$n_1=n_2=n_3=n_4=5$\\  $m_1=m_2=m_3=m_4=5$  \end{tabular} & $p_1=0.4$, $p_2=0.2$, and $p_3=0.3$ & 0.055 & $10^3$ & $10^3$  \\ 
 \hline
 
 2 & \begin{tabular}{@{}l@{}@{}}  $K=3$\\$n_1=n_2=n_3=n_4=10$\\  $m_1=m_2=m_3=m_4=10$  \end{tabular} & $p_1=0.4$, $p_2=0.2$, and $p_3=0.3$ & 0.052 & $10^3$ & $10^3$  \\ 
 \hline
 
  3 & \begin{tabular}{@{}l@{}@{}}  $K=3$\\$n_1=n_2=n_3=n_4=20$\\  $m_1=m_2=m_3=m_4=20$  \end{tabular} & $p_1=0.4$, $p_2=0.2$, and $p_3=0.3$ & 0.049 & $10^3$ & $10^3$  \\ 
 \hline
\end{tabular}
\caption{Simulation of type I error}\label{tab:err}
\end{table}
In order to simulate the type I error, the ground truth probability is set to be $p_1=0.4$, $p_2=0.2$, and $p_3=0.3$. For each of the three experiments in Table\ref{tab:err}, $N_{rep}=10^3$ many Freq Distribution Tables were generated with the table structure as in Table\ref{tab:fdt}. For each of the generated tables, the testing procedure as described in Table\ref{tab:step} was adopted. An Empirical Error rate (Err) was then calculated on the proportion of tests rejected with $\alpha=0.05$ from the $10^3$ tables. Based on the Err results of all three experiments, the empirical type I error of the test was consistent with what it was designed to be.

\section{Real Data Analysis and Conclusion}
Using the real data we compared the sub-sample pattern of observations from treatment PAC with the pattern in treatment PACGA for line "Mo298". 

Table\ref{exp2.sam} shows the sub-sample sizes for the PAC group, as well as for the PACGA group. The designed sub-sample size is 5 for all sub-samples. Because of missing values, not all sub-samples are 5. There are two sub samples in both PAC and PACGA that have all 5 missing (the sixth sub-sample and the eighth).

\begin{table}[h]
\centering
\begin{tabular}{ |c|c|c|c|c|c|c|c|c|  } 
 \hline
 PCA &5 &  5& 5  &3 &1 &0 &  2& 0 \\ 
 \hline
 PACGA&3 &  3& 4  &5 &1&0 &  1& 0   \\ 
 \hline
\end{tabular}
\caption{Freq table of real data PAC and PACGA}\label{exp2.sam}
\end{table}

Table\ref{exp2.cont.table} is based on the formulation in Table\ref{tab:fdt}. Because of the missing values in the entire sub samples, only four sets of comparisons were conducted, which results the following $8\times 2$ table.

\begin{table}[h]
\centering
\begin{tabular}{ |c|c|c|c|c|c|c|c|  } 
 \hline
  12.5 & 12.5& 6.5  &18.5 &13&2 &  0& 3   \\ 
 \hline
 5 &  4& 9  &3 &5 &15 &  2& 3 \\
 \hline
\end{tabular}
\caption{Freq table for real data}\label{exp2.cont.table}
\end{table}

Based on the computation steps mentioned in Table\ref{tab:step}, the test statistic is 31.598, and the $5\%$ bootstrapped critical value is 25.757 based on $N_b=10^3$. The computed p-value is 0.0194.

In summary, this work is motivated from a biological experiment of split-plot design, for the purpose of comparing the changing patterns in seed weight subject to increasing levels of stress when there were multiple populations to compare. We formalized the question into a nonparametric two sample comparison problem for changes among sub samples using U-statistics, with allowance for missing values.  The U-statistics were then used in a Chi-square type test statistic for hypothesis testing. Bootstrapped p-values were obtained. It was proven that the distribution of the simulated sample can be independent with the observed samples given certain summary statistics. Simulation results suggested that the test was consistent.   

\section*{Acknowledgments}
The motivating biology experiment in this project was partially supported by the National Research Initiative Competitive Grant no. 2009-35100-05028 from the USDA National Institute of Food and Agriculture.


\makeatletter   
 \renewcommand{\@seccntformat}[1]{APPENDIX~{\csname the#1\endcsname}.\hspace*{1em}}
 \makeatother

\makeatletter   
 \renewcommand{\@seccntformat}[1]{APPENDIX~{\csname the#1\endcsname}.\hspace*{1em}}
 \makeatother

\mbox{}\vspace*{1ex}
\mbox{}

\bibliographystyle{ECA_jasa}
\bibliography{ofiv4}

\end{document}